\documentclass[11pt]{article}

\usepackage{a4wide}
\usepackage{authblk}
\usepackage{amsmath,amssymb}
\usepackage{amsthm}
\usepackage{graphicx}

\pdfoutput=1

\usepackage{thmtools}

\newtheorem{theorem}{Theorem}
\newtheorem{lemma}{Lemma}

\theoremstyle{definition}

\newtheorem{fact}{Fact}

\usepackage{array}
\usepackage{ascmac}
\usepackage{booktabs}
\usepackage{wrapfig}
\usepackage{url}
\usepackage{here}

\newcommand{\occ}{\mathsf{occ}}
\newcommand{\MUS}{\mathsf{MUS}}

\newcommand{\mus[1]}{\mathrm{M}_{#1}}

\title{On the number of MUSs crossing a position}
\date{}
\author[1]{Hiroto~Fujimaru}
\author[2]{Takuya~Mieno}
\author[3]{Shunsuke~Inenaga}

\affil[1]{Department of Information Science and Technology, Kyushu University}
  \affil[ ]{\texttt{fujimaru.hiroto.134@s.kyushu-u.ac.jp}}

\affil[2]{Department of Computer and Network Engineering, University~of~Electro-Communications}
  \affil[ ]{\texttt{tmieno@uec.ac.jp}}

\affil[3]{Department of Informatics, Kyushu University}
 \affil[ ]{\texttt {inenaga.shunsuke.380@m.kyushu-u.ac.jp}}

\begin{document}
\maketitle

\begin{abstract}
  A string $w$ is said to be a \emph{minimal unique substring} (\emph{MUS}) of a string $T$ if $w$ occurs exactly once in $T$,
  and any proper substring of $w$ occurs at least twice in $T$.
  It is known that the number of MUSs in a string $T$ of length $n$
  is at most $n$,
  and that the set $\MUS(T)$ of all MUSs in $T$ can be computed in $O(n)$ time~[Ilie and Smyth, 2011].
  Let $\MUS(T,i)$ denote the set of MUSs 
  that contain a position $i$ in a string $T$.
  In this short paper, we present matching $\Theta(\sqrt{n})$ upper and lower bounds
  for the number $|\MUS(T,i)|$ of MUSs containing a position $i$ in a string $T$ of length $n$.
\end{abstract}

\section{Introduction}\label{sec:intro}

A string $w$ is said to be a \emph{minimal unique substring} (\emph{MUS}) of a string $T$ if $w$ occurs exactly once in $T$, and any proper substring of $w$ occurs at least twice in $T$.
Let $\MUS(T)$ denote the set of MUSs in $T$.
It is known~\cite{IlieS11} that the number $|\MUS(T)|$ of MUSs
in any string $T$ of length $n$ is at most $n$,
and that this bound is tight (e.g. consider a string consisting of $n$ distinct characters).
This number $|\MUS(T)|$ can also be bounded by other parameters
that relate to compression:
For any string $T$ whose \emph{run-length encoding} (\emph{RLE}) size is $m$, $|\MUS(T)| \leq 2m-1$ holds~\cite{MienoIBT16}.
Also, if the number of edges of the \emph{compact directed acyclic word graph} (\emph{CDAWG})~\cite{BlumerBHME87} of a string $T$ is $e$, then $|\MUS(T)| \leq e$ holds~\cite{InenagaMAFF24}.
Computing $\MUS(T)$ is important preprocessing for
the \emph{shortest unique substring} (\emph{SUS}) queries~\cite{AbedinKT20},
which are motivated by bioinformatics applications~\cite{HauboldPMW05,PeiWY13}.

Since each element $w \in \MUS(T)$ can be represented by a unique pair $(i,j)$
such that $w = T[i..j]$,
the set $\MUS(T)$ can be represented by $O(n \log n)$ bits of space,
or in $O(\min\{m,e\} \log n)$ bits of space.
Ilie and Smyth~\cite{IlieS11} showed how to compute
$\MUS(T)$ in $O(n)$ time with $O(n\log n)$ bits of working space.
There are also other space-efficient algorithms for computing $\MUS(T)$~\cite{MienoIBT16,BelazzouguiCGPR15,MienoKNIBT20,NishimotoT21,InenagaMAFF24}.

In this paper, we are interested in the following:
``How many MUSs in $\MUS(T)$ can contain the same position in $T$?''
Namely, we evaluate the worst-case bounds for the number $|\MUS(T,i)| = |\{T[k..j] \in \MUS(T) \mid k \leq i \leq j\}|$ of MUSs that contain position $i$ in $T$.
We prove that $|\MUS(T,i)| = O(\sqrt{n})$ holds for any string $T$ of $n$ and any position $i$~($1 \leq i \leq n$), and present a family of strings for which $|\MUS(T,i)| = \Omega(\sqrt{n})$ holds. Hence our bounds for $|\MUS(T,i)|$ are tight.

While the new combinatorial properties of MUSs presented in this paper
are interesting in their own right,
this study is also well motivated for applications including
design of a dynamic data structure
for storing the set $\MUS(T)$ of MUSs.
Namely, our results can be seen as the first step toward a representation of MUSs allowing for sublinear-time updates.
Another example of our motivation is to analyze the \emph{sensitivity}~\cite{AkagiFI23}
of MUSs,
which is the worst-case increase in the number $|\MUS(T)|$ of MUSs
after performing a single-character edit operation in $T$.
Our $\Omega(\sqrt{n})$ lower bound instance for $|\MUS(T,i)|$
immediately leads to an $\Omega(\sqrt{n})$ lower bound for the sensitivity of $|\MUS(T)|$,
and our $O(\sqrt{n})$ upper bound for $|\MUS(T,i)|$ suggests that
the sensitivity of $|\MUS(T)|$ is also $O(\sqrt{n})$.

\section{Preliminaries}

 Let $\Sigma$ be an alphabet.
An element of $\Sigma$ is called a character.
An element of $\Sigma^\ast$ is called a string.
The length of a string $T$ is denoted by $|T|$.
The empty string $\varepsilon$ is the string of length 0.
For a string $T = xyz$, $x$, $y$, and $z$ are called
a \emph{prefix}, \emph{substring}, and \emph{suffix} of $T$, respectively.
For a string $T$ of length $n$,
$T[i]$ denotes the $i$th character of $T$
for $1 \leq i \leq n$,
and $T[i..j]$ denotes
the substring of $T$ that begins at position $i$ and ends at position $j$
for $1 \leq i \leq j \leq n$.
For convenience, let $T[i..j] = \varepsilon$ for $i > j$.
We say that string $w$ \emph{occurs} in a string $T$
iff $w$ is a substring of $T$.
An interval $[i.. j]$ is said to be an \emph{occurrence} of $w$ in $T$
if $w = T[i.. j]$.
We may identify an occurrence $[i.. j]$ of $w$ with the string $w = T[i.. j]$
if there is no confusion.
Let $\occ_T(w)$ denote
the number of occurrences of $w$ in $T$.
For convenience, let $\occ_T(\varepsilon) = |T|+1$.
For substring $T[i.. j]$ of $T$ and position $1 \le p \le |T|$,
we say that $T[i.. j]$ contains $p$ if $i \le p \le j$.
An integer $p \geq 1$ is said to be a \emph{period}
of a string $T$ if $T[i] = T[i+p]$ holds for all $1 \leq i \leq n-p$.
We use the following fact, which can be shown with the \emph{periodicity lemma}~\cite{fine1965uniqueness}:
\begin{fact}\label{fact:threeoverlap}
  If a string $S$ occurs at three positions $i, j$, and $k$ in $T$ with $i < j < k \le i+|S|-1$,
  then $\gcd(j-i, k-j)$ is a period of $T[i.. k+|S|-1]$, and hence of $S$.
  Thus, $T[j-1] = T[k-1]$ holds.
\end{fact}

A substring $w$ of string $T$
is called a \emph{unique} substring in $T$ if $\occ_T(w) = 1$,
and it is called a \emph{repeat} in $T$ if $\occ_T(w) \geq 2$.
A unique substring $w$ in $T$
is called a \emph{minimal unique substring} (\emph{MUS}) in $T$
if $\occ_T(w[1..|w|-1]) \geq 2$ and $\occ_T(w[2..|w|]) \geq 2$.
Let $\MUS(T)$ denote the set of MUSs in string $T$.
Let $\MUS(T,i) = \{T[k..j] \in \MUS(T) \mid k \leq i \leq j\}$ denote the set of MUSs that contain position $i$ in $T$.
By definition, MUSs do not nest, namely,
for two distinct MUSs $T[i..j], T[i'..j'] \in \MUS(T)$,
$[i.. j]$ is not a subinterval of $[i'.. j']$, and vice versa.

\section{Bounds for maximum of $|\MUS(T,i)|$}
In this section, we show the following:
\begin{theorem} \label{theo:upper_bound_crossing_MUS}
  For any string $T$ of length $n$ and any position $i$ in $T$, 
  $|\MUS(T,i)| \in O(\sqrt{n})$ holds.
  Also, this upper bound is tight.
\end{theorem}

\subsection{Upper bound}
In this subsection, we fix a string $T$ of length $n$ arbitrarily.
We first prove our key lemma
to our $O(\sqrt{n})$ upper bound for $|\MUS(T,i)|$.
Suppose that $\mus[1] = [i_1 .. i_1+|\mus[1]|-1], \mus[2] = [i_2 .. i_2+|\mus[2]|-1], \mus[3] = [i_3 .. i_3+|\mus[3]|-1]$ are MUSs
that contain the position $i$ and occur in this order.
Let $a_k = T[i_k]$ be the first character of $\mus[k]$
and $s_k = T[i_k+1 .. i_k+|\mus[k]|-1]$ be the rest for each $k\in\{1,2,3\}$.
Since $\mus[3]$ is the rightmost MUS among three MUSs $\mus[1], \mus[2], \mus[3]$,
the overlap of them can be written as $a_3u$ for some string $u \in \Sigma^*$.
Let $q = T[i_1+1..i_3]$ and $p = T[i_2+1..i_3]$
be non-empty prefixes of $s_1$ and $s_2$
that immediately precede the occurrence $[i_3+1.. i_3+|u|]$ of $u$ respectively~(see also Fig.~\ref{fig:u_NOT_overlapped}).
Further let $r = T[i_2+|pu|+1 .. i_2+|M_2|-1]$ be the suffix of $s_2$
that immediately follows the occurrence $[i_2 + |p|+1.. i_2 + |pu|]$ of $u$.
By the definition of MUSs, 
$s_k$ has an occurrence other than $[i_k+1.. i_k+|s_k|]$ for each $k\in\{1,2,3\}$.
\begin{figure}[tb]
  \centering
  \includegraphics[keepaspectratio,width=\linewidth]{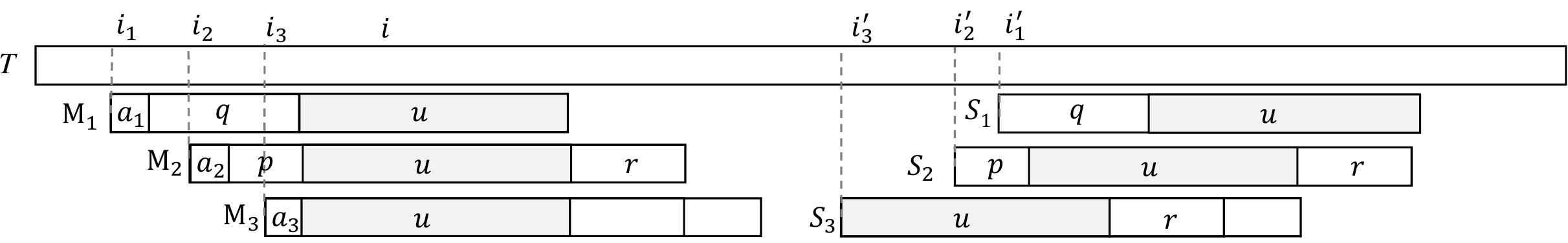}
  \caption{Illustration for Lemma~\ref{lem:u_NOT_overlapped}.
  This depicts one of several possible mutual locations of $S_1,S_2,S_3$.
  }
  \label{fig:u_NOT_overlapped}
\end{figure}
For $s_1, s_2$, and $s_3$, we have following lemma:
\begin{lemma}[key lemma] \label{lem:u_NOT_overlapped}
  Let $S_1 = [i'_1..i'_1+|s_1|-1]$, $S_2 = [i'_2..i'_2+|s_2|-1]$, $S_3 = [i'_3..i'_3+|s_3|-1]$ be any occurrences of 
  $s_1, s_2, s_3$ in $T$ such that
  $i'_1 \neq i_1+1$, $i'_2 \neq i_2+1$, and $i'_3 \neq i_3+1$, respectively.
  Then, for the three occurrences
$U_1=[i'_1+|q| .. i'_1+|q|+|u|-1]$,
$U_2=[i'_2+|p| .. i'_2+|p|+|u|-1]$,
  $U_3=[i'_3..i'_3+|u|-1]$ of
  $u$ in $S_1, S_2, S_3$,
  at least two of them do not overlap (see Fig.~\ref{fig:u_NOT_overlapped}).
\end{lemma}
\begin{proof}
  For the sake of contradiction,
  we assume that the three occurrences of $u$ overlap.
  Let $a'_2 = T[i'_2-1]$ and $a'_3 = T[i'_3-1]$ be the characters that precede $S_2$ and $S_3$ respectively.
  By the definition of MUSs, $a_2 \neq a'_2$ and $a_3 \neq a'_3$ hold.
  While the character immediately before the occurrence $U_3$ of $u$ is $a'_3$,
the characters immediately before the occurrences $U_1$ and $U_2$ of $u$ must be $T[i_3] = a_3$.
For the overlapped three occurrences of $u$, due to Fact~\ref{fact:threeoverlap}, the characters immediately
  before two occurrences except the leftmost one are the same. 
  This implies that if $U_3$ is not the leftmost among $U_1, U_2, U_3$,
  then $a_3 = a'_3$, which is contradiction.
  Thus, $U_3$ is the leftmost occurrence among the three occurrences of $u$.

Let $S'_1 = T[i'_1 + |q| - |p| .. i'_1+|s_1|-1]$ be the occurrence of $pu$ as a suffix of $S_1$, 
  and $i''_1 = i'_1 + |q| - |p|$ be the starting position of $S'_1$.
  Let $P_1 = [i''_1 .. i''_1+|p|-1]$ and $P_2 = [i'_2 .. i'_2 + |p|-1]$
  be the occurrences of $p$ as prefixes of $S'_1$ and $S_2$ respectively.
  It holds that $i''_1 \neq i'_2$ since $T[i''_1-1] = a_2 \neq a'_2 = T[i'_2-1]$.
In the following, we consider the four cases depending on $i''_1, i'_2, i'_3$, 
  which are the starting positions of $S'_1,S_2,S_3$ respectively; 
  Case 1. $i'_2 < i'_3$ or $i''_1 < i'_3$;
  Case 2. $i'_3 < i'_2$ and $i'_3 < i''_1$;
  Case 3. $i'_3 = i''_1 < i'_2$;
  Case 4. $i'_3 = i'_2 < i''_1$.

  \textbf{Case 1: $i'_2 < i'_3$ or $i''_1 < i'_3$.}
  We further consider the sub-case where $i'_2 < i''_1$.
  The other sub-case can be shown similarly.
  By assumptions, $i'_2 < i'_3$ holds, and thus, $a'_3u$ that occurs at position $i'_3-1$ is a substring of $pu$~(see also Fig.~\ref{fig:u_case1}).
  Since two occurrences $P_1U_1, P_2U_2$ of $pu$ overlap, $pu$ has period $i''_1 - i'_2$.
  Hence, its substring $a'_3u$ also has period $i''_1 - i'_2$.
  The periodicity implies that $a_3 = u[i''_1-i'_2] = a'_3$, a contradiction.
\begin{figure}[tb]
    \centering
    \includegraphics[keepaspectratio,width=0.65\linewidth]{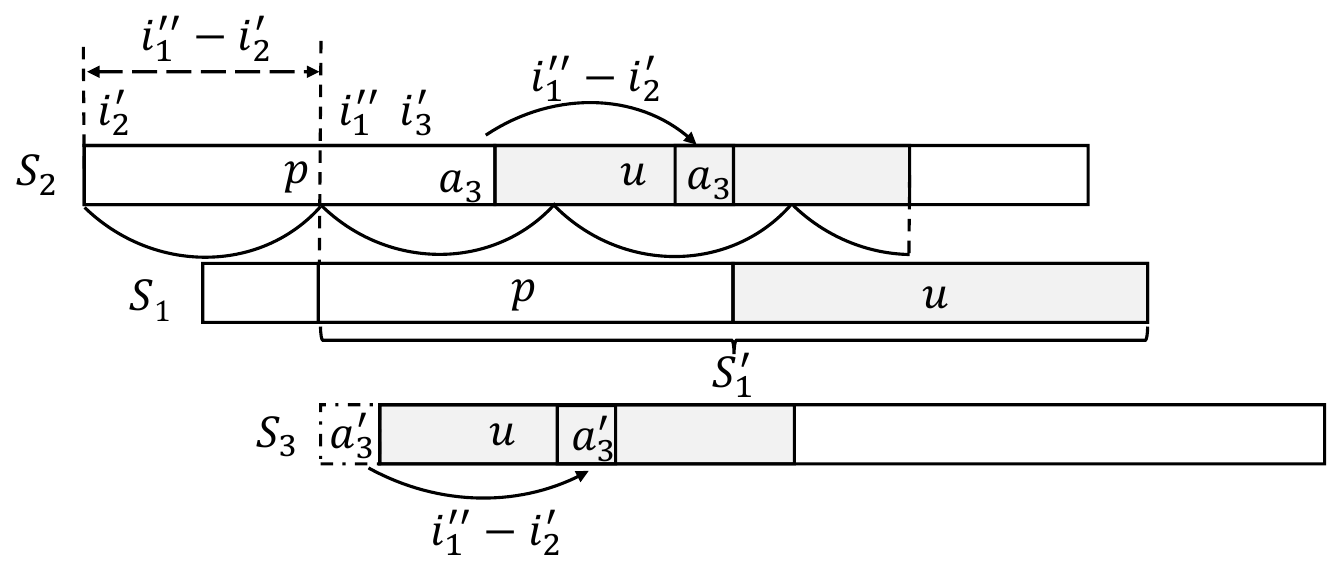}
    \caption{Illustration for a contradiction in Case 1 of the proof of Lemma~\ref{lem:u_NOT_overlapped}.}
\label{fig:u_case1}
  \end{figure}

  \textbf{Case 2: $i'_3 < i'_2$ and $i'_3 < i''_1$.}
  We further consider the sub-case where $i'_3 < i'_2 < i''_1$.
  The other sub-case can be shown similarly.
  As in Case 1,
$pu$ has period $i''_1 - i'_2$.
  Also, its substring $u$ has period $i''_1 - i'_2$.
  Since $i'_3 < i'_2 < i''_1$ holds and three occurrences  $U_1, U_2, U_3$ overlap, $i'_3 \leq i'_2-1 < i''_1-1 < i'_3+|u|-1$.
  This implies that both positions $i'_2-1$ and $i''_1-1$ are within occurrence $U_3$ of $u$~(see also Fig.~\ref{fig:u_case2}).
Since $u$ has period $i''_1 - i'_2$, it holds that $a'_2 = T[i'_2-1] = T[i'_2-1+(i''_1-i'_2)] = T[i''_1-1] = a_2$, a contradiction.
\begin{figure}[tb]
  \centering
  \includegraphics[keepaspectratio,width=\linewidth]{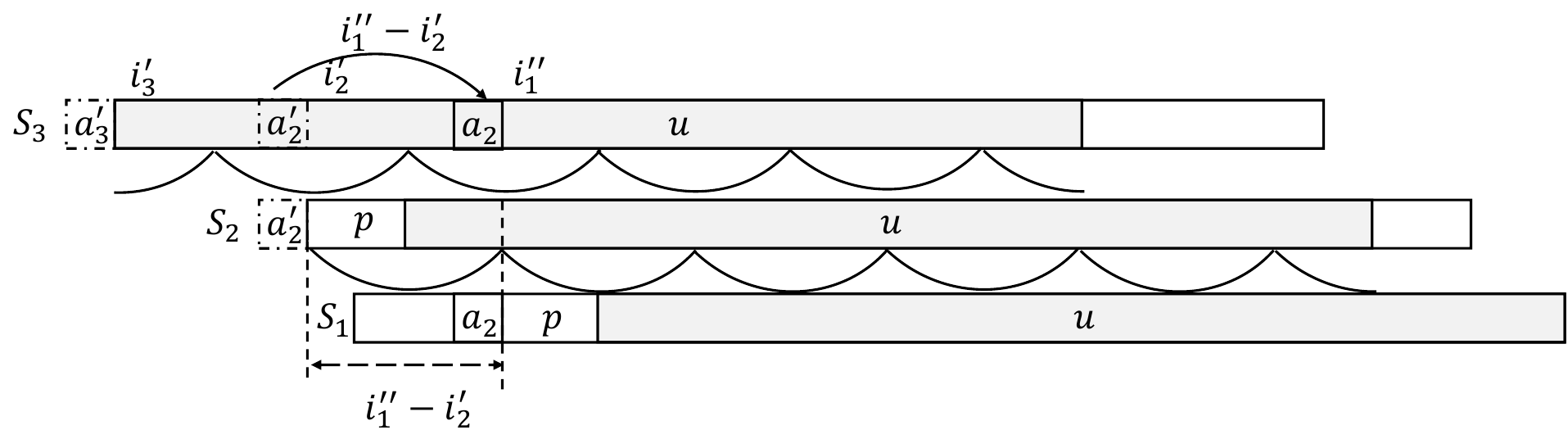}
  \caption{Illustration for a contradiction in Case 2 of the proof of Lemma~\ref{lem:u_NOT_overlapped}.}
\label{fig:u_case2}
  \end{figure}

  \textbf{Case 3: $i''_1 = i'_3 < i'_2$.}
Let us consider periods of $u$ in detail.
Since three overlapping occurrences $U_3, U_1$, and $U_2$ of $u$ appear in this order,
  $u$ has period $\gcd((i''_1+|p|)-i'_3, i'_2+|p|-(i''_1+|p|)) = \gcd(|p|, i'_2-i''_1)$ due to Fact~\ref{fact:threeoverlap}.
  Let $\alpha = \gcd(|p|, i'_2-i''_1)$.
Next, we consider periods of $r$.
  Let $R_2 = [i'_2+|pu| .. i'_2+|s_2|-1]$ and $R_3 = [i'_3+|u| .. i'_3 + |ur| - 1]$ be occurrences of $r$
  that immediately follow the occurrences $U_2$ and $U_3$ of $u$ respectively.
  We consider two sub-cases 3-1 and 3-2 depending on the length of $r$:
\textbf{Case 3-1.}~If $|r| > i'_2-i''_1+|p|$, then
  $U_2$ and $R_3$ overlap, and the length of their intersection is $i'_2-i''_1+|p|$~(see also Fig.~\ref{fig:u_case3}).
  Also, the length of the overlap of occurrences $R_3$ and $R_2$ of $r$ is $|r|-(i'_2-i''_1+|p|)$,
  and thus $r$ has period $i'_2-i''_1+|p|$.
  By the definition of $\alpha = \gcd(|p|, i'_2-i''_1)$, $i'_2-i''_1+|p|$ is a multiple of $\alpha$.
  Thus, the overlap of $U_2$ and $R_3$ forms an integer power of a string of length $\alpha$,
  and $r$ also has period $\alpha$.
  To summarize, string $ur$ also has period $\alpha$,
  since both $u$ and $r$ have period $\alpha$ and their overlap is a multiple of $\alpha$.
\textbf{Case 3-2.}~If $|r| \le i'_2-i''_1+|p|$, then
  $r$ is a prefix of the length-$(i'_2-i''_1+|p|)$ suffix of $u$, which is an integer power of some string of length $\alpha$.
  Therefore, similar to the above, $ur$ has period $\alpha$.

In both sub-cases, $ur$ has period $\alpha$.
Next, let us consider $x = T[i''_1..i'_2+|s_2|-1]$. 
  Since $ur$ is a border of $x$
and $|x| -  |ur| = i'_2-i''_1+|p|$ is a multiple of $\alpha$,
  $x$ has period $\alpha$.
Then, $s_2 = T[i'_2 .. i'_2+|s_2|-1] = T[i''_1 .. i''_1+|s_2|-1]$ holds.
  Since $T[i''_1-1] = a_2$, $T[i''_1-1..i''_1+|s_2|-1] = a_2s_2$ ,which is identical to the second MUS $T[i_2..i_2 + |\mus[2]|-1]$. Since a MUS must be unique in $T$, $i''_1-1 = i_2$ holds.
Then, considering the starting position $i'_1$ of $S_1$,
  we have $i'_1 = i''_1+|p|-|q| = i_2+1 +|p|-|q| = i_1+1$, which contradicts the definition of $i'_1$.
\begin{figure}[tb]
    \centering
    \includegraphics[keepaspectratio,width=\linewidth]{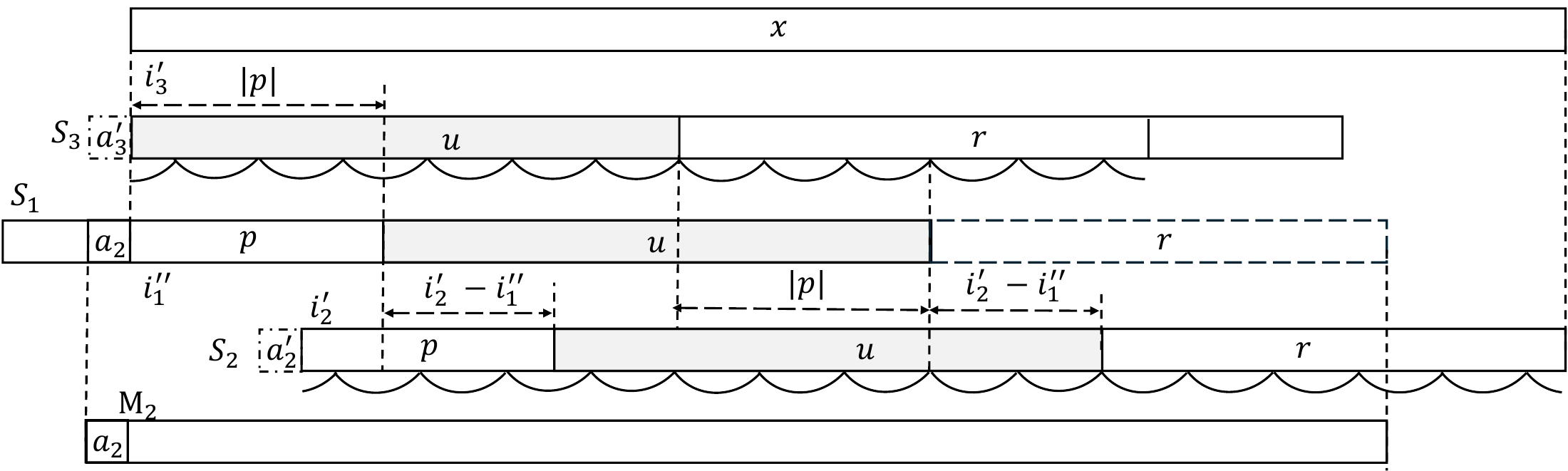}
    \caption{Illustration for Case 3 of the proof of Lemma~\ref{lem:u_NOT_overlapped}.}
\label{fig:u_case3}
  \end{figure}

  \textbf{Case 4: $i'_2 = i'_3 < i''_1$.}
  Since two occurrences $U_3$ and $U_2$ of $u$ overlap, $u$ has period $|p|$~(see also Fig.~\ref{fig:u_case4}).
Also, $T[i''_1 + |p|-1] = a_3$ holds
  since the character that precedes the suffix $u$ of $s_1$ is $a_3$.
Since $U_3$ and $U_1$ overlap and $i'_3 < i''_1$,
  $T[i''_1-1 .. i''_1+|p|-1] = a_2p$ is a substring of $u$, which has period $|p|$.
  Thus, $a_2 = T[i''_1-1] = T[i''_1+|p|-1] = a_3$ holds.
Here, one of $s_2$ and $s_3$ is a prefix of the other since $i'_2 = i'_3$ holds.
  Therefore, one of $a_2s_2$ and $a_3s_3 = a_2s_3$ is a prefix of the other,
  which contradicts that MUSs $a_2s_2$ and $a_3s_3$ are unique in $T$.
  \end{proof}

\begin{figure}[h]
  \centering
  \includegraphics[keepaspectratio,width=0.9\linewidth]{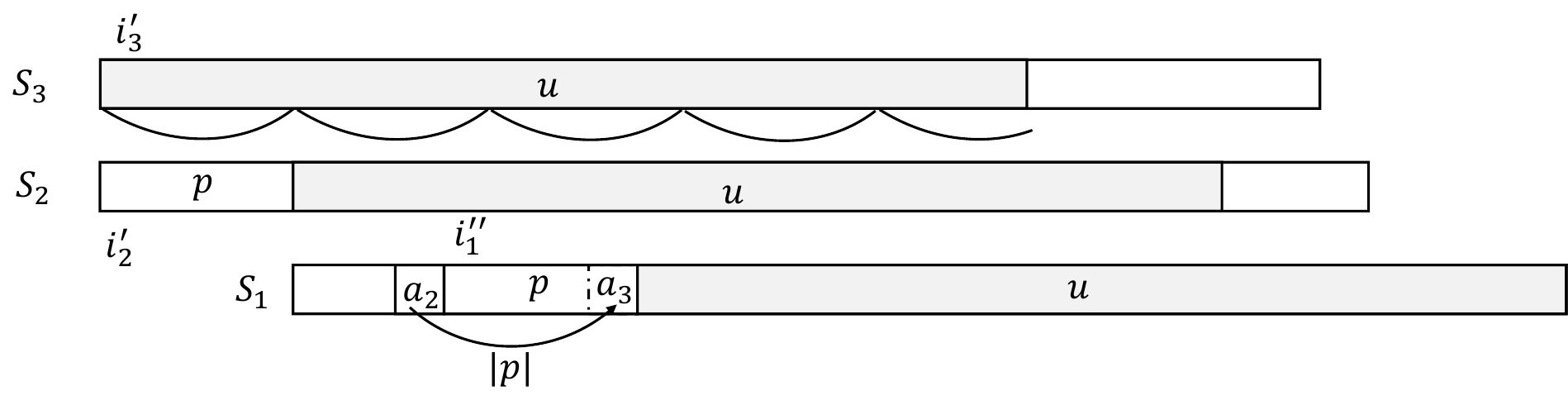}
  \caption{Illustration for a contradiction in Case 4 of the proof of Lemma~\ref{lem:u_NOT_overlapped}.}
\label{fig:u_case4}
  \end{figure}

Based on the key lemma, we estimate an upper bound for $|\MUS(T, i)|$.
Suppose that $h+1$ MUSs contain the same position $i$ in $T$.
Let $\mus[1], \mus[2], \dots , \mus[h+1]$ be such MUSs from left to right.
We focus only on the first $h$ MUSs.
For each $1 \le k \le h$, let $s_k = \mus[k][2..|\mus[k]|]$ be the suffix of $\mus[k]$ of length $|\mus[k]|-1$.
We \emph{mark} the character at position $i$ on $s_k$~(see Fig.~\ref{fig:MUS_marker}).
Let $S_k$ be an occurrence of $s_k$ that is not the suffix of $\mus[k]$ for each $1\le k \le h$.
We denote by $i_k$ the marked position on the occurrence $S_k$ of $s_k$.
Further let $I = \{i_k \mid 1\le k \le h\}$ be the set of marked positions excluding $i$.
We define a function $f: \{1,\ldots, h\} \to \{1,\ldots, h\}$ such that
$f(x) = k$ iff $i_k$ is the $x$th smallest value in $I$.
In other words,
$f(x) = k$ iff $x = |\{i_\ell \in I \mid i_\ell \le i_k\}|$.
\begin{figure}[tb]
  \centering
  \includegraphics[keepaspectratio,width=\linewidth]{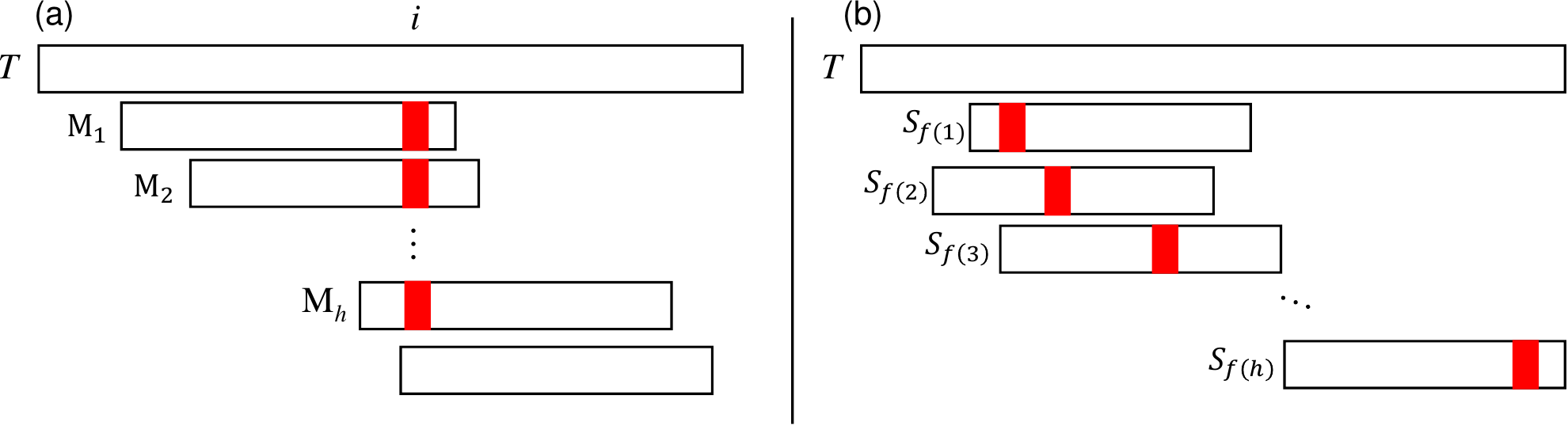}
  \caption{(a) These are $h+1$ MUSs that contain position $i$. Markers are indicated in red.
    (b)
Strings $S_{f(1)}, S_{f(2)}, \ldots, S_{f(h)}$ are sorted by their marked positions.
  }
  \label{fig:MUS_marker}
\end{figure}
For each $1 \le x \le h-2$, we have the following:
\begin{figure}[tb]
  \centering
  \includegraphics[keepaspectratio,width=\linewidth]{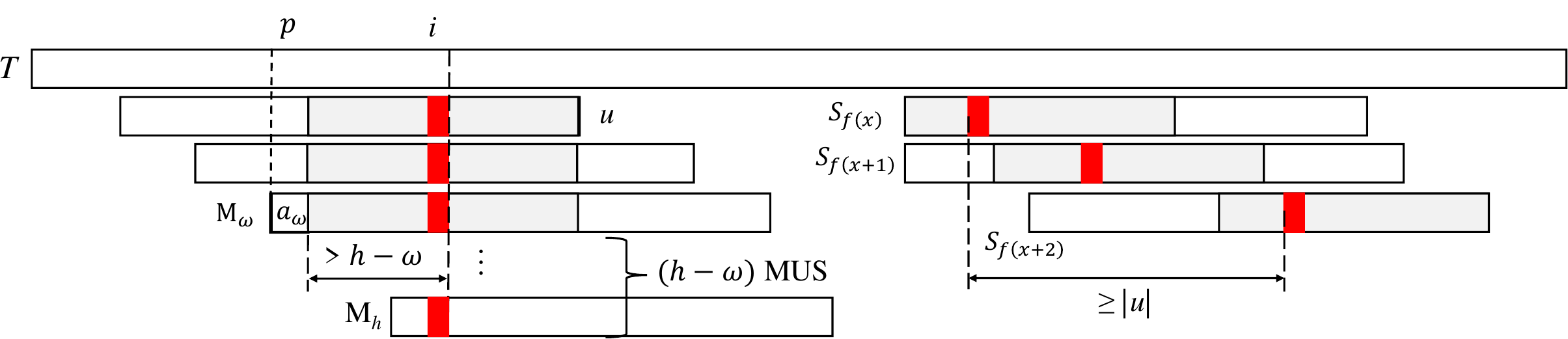}
  \caption{Illustration for the proof of Lemma~\ref{lem:ik_interval}.}
  \label{fig:ik_interval}
\end{figure}
\begin{lemma} \label{lem:ik_interval}
$i_{f(x+2)} - i_{f(x)} > h - \omega$ holds
  where $\omega = \max\{f(x), f(x+1), f(x+2)\}$.
\end{lemma}
\begin{proof}
  Let $au$ be the overlap of $\mus[f(x)], \mus[f(x+1)]$ and $\mus[f(x+2)]$ where $a \in \Sigma$.
  By Lemma~\ref{lem:u_NOT_overlapped}, 
at least two of the three occurrences of $u$ in $S_{f(x)}, S_{f(x+1)}$, and $S_{f(x+2)}$ do not overlap.
  Thus, marked positions in $S_{f(x)}$ and $S_{f(x+2)}$ are separated by a distance of at least $|u|$~(see also Fig.~\ref{fig:ik_interval}).
  Namely, $i_{f(x+2)} - i_{f(x)} \geq |u|$ holds
Next, we analyze the length of $u$.
  Let $p$ be the starting position of $\mus[\omega]$.
  Since MUSs do not nest and
there are $h+1-\omega$ MUSs that contain $i$ whose starting positions are greater than $p$,
$i-p \geq h-\omega+1$ holds.
  Therefore, $|u| \ge i-p > h-\omega$.
  To summarize, we have 
  $i_{f(x+2)} - i_{f(x)} \ge |u| > h - \omega$.
\end{proof}

Let $\omega_x = \max\{f(x), f(x+1), f(x+2)\}$ for each $1 \le x \le h-2$.
Applying Lemma~\ref{lem:ik_interval} for all $1 \leq x \leq h-2$, we obtain
$\sum_{1\leq x \leq h-2}(h-\omega_x) < \sum_{1\leq x \leq h-2}(i_{f(x+2)}-i_{f(x)}) = i_{f(h)} + i_{f(h-1)} - i_{f(2)} - i_{f(1)} < 2n$.
To estimate a lower bound of $\sum_{1\leq x \leq h-2}(h-\omega_x) = h(h-2) - \sum_{1\leq x \leq h-2}\omega_x$,
we consider an upper bound of $\sum_{1\leq x \leq h-2}\omega_x$.
Since $\omega_x$ is defined as the minimum among three function values,
each value $q$ with $1 \le q \le h$ can be selected as $\omega_x$ at most three times.
Thus, $\sum_{1 \le x \le h-2}\omega_x \le 3\sum_{h - \lceil h/3 \rceil + 1 \leq x \leq h}x < 5h^2/6+3h+3/2$ holds. 
Combining the inequalities, we have
$2n > \sum_{1 \leq x \leq h-2}(h-\omega_x) = h(h-2) - \sum_{1 \leq x \leq h-2}\omega_x > h(h-2) - (5h^2/6+3h+3/2) = h^2/6 + h -3/2$,
which implies $h+1 \in O(\sqrt{n})$.
This concludes the proof of $|\MUS(T, i)| \in O(\sqrt{n})$.

\subsection{Lower bound} \label{sec:lowerbound}

\begin{lemma}\label{lem:lowerbound}
  There exists an infinite family of strings $T_m$ that satisfies \linebreak
  $|\MUS(T_m, p)| \in \Omega(\sqrt{|T_m|})$
  for a position $p$ in $T_m$.
\end{lemma}
\begin{proof}
Consider the string
  $$T_m = ab^{2m}ab^{2m+2} S_{m,1}  S_{m,2} \cdots S_{m,k} \cdots S_{m,m-1},$$
  where
  $S_{m,k} = ab^kab^{2m-k}$ for each $1\leq k\leq m-1$.
  Let $p = 2m+4$. See also Fig.~\ref{fig:lowerex} for a concrete example.

  For each $2 \leq i \leq m-1$,
  consider the substring
  $$T_m[p-(2+i)..p+(2m-1-i)] = b^i a b^{2m-i+1}$$
  of $T_m$. 
  It is clear that each $T_m[p-(2+i)..p+(2m-1-i)]$ contains position $p$.

  Now we show that each $b^i a b^{2m-i+1}$ is a MUS of $T_m$.
  First, $b^i a b^{2m-i+1}$ is unique in $T_m$.
  Second, $b^{i-1} a b^{2m-i+1}$ is a repeat in $T_m$
  since $S_{m,i-1}[2..2m+2]= b^{i-1} a b^{2m-i+1}$.
  Third, $b^i a b^{2m-i}$ is a repeat in $T_m$
  since $S_{m,i}[2..2m+2] = b^i a b^{2m-i}$.
  Thus, $\{b^i a b^{2m-i+1} \mid 2 \leq i \leq m-1\} \subset \MUS(T_m, p)$.
  Since $|T_m| = 2m^2 + 4m + 2 \in \Theta(m^2)$, $|\MUS(T_m, p)| = \Omega(\sqrt{|T_m|})$ holds.
\end{proof}

\begin{figure}[tb]
  \centering
  \includegraphics[keepaspectratio,width=0.9\linewidth]{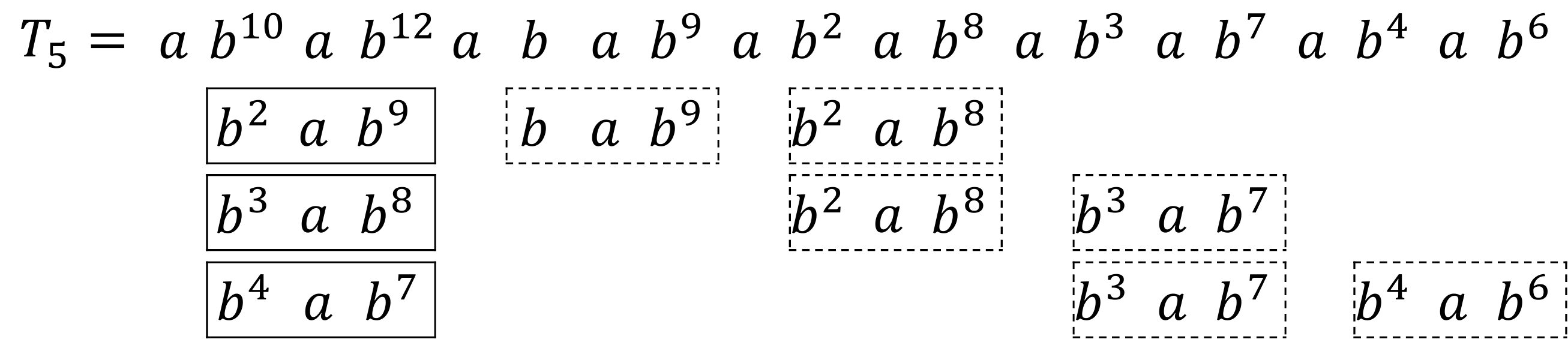}
  \caption{An example of $T_m$ for $m=5$ in the proof of Lemma~\ref{lem:lowerbound}.
  There are $m-2 = 3$ MUSs $b^2 a b^9$, $b^3 a b^8$, and $b^4 a b^7$ that contain the position $2m+4 = 14$.}
\label{fig:lowerex}
\end{figure}

To summarize, we have shown Theorem~\ref{theo:upper_bound_crossing_MUS}.

\section{Conclusions and further work} \label{sec:conclusions}

In this paper, we presented a tight $\Theta(\sqrt{n})$ bound on the number
$|\MUS(T,i)|$ of the MUSs that contains a position $i$ in any string $T$ of length $n$.

One of our future work includes analysis of the sensitivity of MUSs, which is the worst-case increase in the number $|\MUS(T)|$ of MUSs after performing a single-character edit operation on $T$.
Our $\Omega(\sqrt{n})$ lower bound instance for $|\MUS(T,i)|$ in Section~\ref{sec:lowerbound}
already gives us an $\Omega(\sqrt{n})$ lower bound for the sensitivity of MUSs
in both additive and multiplicative senses.
We conjecture that the worst-case additive/multiplicative sensitivity of $|\MUS(T)|$ is also $O(\sqrt{n})$.
Our $O(\sqrt{n})$ upper bound for $|\MUS(T,i)|$
presented in this paper partly contributes to
analyzing an upper bound of the sensitivity of $|\MUS(T)|$,
since the number of new MUSs containing the edit position $i$ is equal to $|\MUS(T,i)|$. 
What remains is to analyze the number 
$|\MUS(T) \setminus (\MUS(T,i) \cup \MUS(T'))|$
of new MUSs which do \emph{not} contain the edited position $i$,
where $T'$ is the string before edit.
 
\section*{Acknowledgements}
This work was supported by JST BOOST Grant Number JPMJBS2406~(HF), and
by JSPS KAKENHI Grant Numbers JP24K20734~(TM), JP23K24808, and JP23K18466~(SI).
We thank Hideo Bannai and Haruki Umezaki for discussions.

\bibliographystyle{abbrv}
\bibliography{ref}

\end{document}